\documentclass[twocolumn,aps,prl]{revtex4-1}

\usepackage{mathpazo}
\usepackage[T1]{fontenc}
\usepackage[latin9]{inputenc}
\usepackage{color}
\usepackage{amsmath}
\usepackage{amssymb}
\usepackage{graphicx}
\usepackage[breaklinks=true,colorlinks=true,citecolor=blue,filecolor=black,linkcolor=red,urlcolor=blue]{hyperref}

\makeatletter

\usepackage{amsthm}

\DeclareMathOperator{\tr}{Tr}
\DeclareMathOperator{\rank}{rank}
\DeclareMathOperator\supp{supp}
\DeclareMathOperator{\spann}{span}
\DeclareMathOperator{\cone}{cone}
\DeclareMathOperator{\conv}{conv}
\DeclareMathOperator{\inter}{int}
\DeclareMathOperator{\cl}{cl}

\newtheorem{lemma}{Lemma}\newtheorem{corol}{Corollary}

\makeatother

\begin{document}

\title{Linear semi-infinite programming approach for entanglement quantification}

\author{Thiago Mureebe Carrijo}
\email{thiagomureebe@gmail.com}
\affiliation{Instituto de Física, Universidade Federal de Goiás, 74.690-900, Goiânia,
Goiás, Brazil}

\author{Wesley Bueno Cardoso}
\affiliation{Instituto de Física, Universidade Federal de Goiás, 74.690-900, Goiânia,
Goiás, Brazil}

\author{Ardiley Torres Avelar}
\affiliation{Instituto de Física, Universidade Federal de Goiás, 74.690-900, Goiânia,
Goiás, Brazil}

\begin{abstract}
We explore the dual problem of the convex roof construction by identifying
it as a linear semi-infinite programming (LSIP) problem. Using the
LSIP theory, we show the absence of a duality gap between primal and
dual problems, even if the entanglement quantifier is not continuous,
and prove that the set of optimal solutions is non-empty and bounded.
In addition, we implement a central cutting-plane algorithm for LSIP
to quantify entanglement between three qubits. The algorithm has global
convergence property and gives lower bounds on the entanglement measure
for non-optimal feasible points. As an application, we use the algorithm
for calculating the convex roof of the three-tangle and $\pi$-tangle
measures for families of states with low and high ranks. As the $\pi$-tangle
measure quantifies the entanglement of W states, we apply the values
of the two quantifiers to distinguish between the two different types
of genuine three-qubit entanglement. 
\end{abstract}

\maketitle

\section{Introduction}

Quantum entanglement is a special type of correlation of quantum systems
with two or more parts, which admit global states that cannot be written
with products of individual states of the parts. The interest in this
phenomenon has origin in its importance in fundamental questions of
quantum mechanics including EPR paradox and nonlocality, in its relationship
with other physical phenomena such as super-radiance, superconductivity
and disordered system, and in technological applications where the
quantum entanglement is a valuable resource for many tasks in the
areas of quantum computing and quantum information \cite{horodecki2009}.

As a consequence, the production, manipulation and quantification
of entanglement are permanent topics of scientific interest. In particular,
the quantification of entanglement can be accomplished using several
different types of entanglement measures that are generally much simpler
to define for pure states than for mixed states. Fortunately, the
construction of a measure for mixed states can be done through the
convex roof of an entanglement monotone \cite{Vidal2000}. However,
the calculation of a convex roof is computationally expensive for
high rank states, with the exception of a few cases whose analytical
solution is known.

Most numerical algorithms for the convex roof calculations work to
find the optimal pure state decomposition of the input state \cite{zyczkowski1999,audenaert2001,rothlisberger2009,cao2010,rothlisberger2012}.
Although this approach can be very efficient for low rank states,
the parameter space of the optimization problem grows fast with the
rank and has a maximal number of parameters of $\sim2n^{3}$ \cite{ryu2012},
where $n$ is the dimension of the system. Also, these methods usually
lack global convergence, which means that they can guarantee only
upper bounds on the optimal value. Another method obtains a sequence
of lower bounds by solving semidefinite programming problems, but
only for measures that are polynomials of expectation values of observables
for pure states \cite{toth2015}. A promising approach, with fewer
optimization parameters, is solving the dual problem of the convex
roof optimization task \cite{brandao2005}. Following this idea,
a minimax algorithm was proposed for the dual problem, which provides
a verifiable globally optimal solution or a lower bound on the convex
roof measure \cite{ryu2012}.

The concept of genuine multipartite entanglement, which applies to
systems with three or more parts, differentiates the correlation among
all subsystems from that restricted to a proper subset of them \cite{horodecki2009,bengtsson2017}.
It is present in many quantum algorithms \cite{bruss2011,pan2017},
cryptographic protocols \cite{cabello2000,epping2017,murta2020}
and quantum phenomena \cite{greenberger1990,hauke2016,smith2016}.
As a resource, it is essential to be able to quantify it, which have
been done by quantifiers like the three-tangle \cite{coffman2000},
its generalization by means of hyperdeterminants \cite{miyake2004},
the $\pi$-tangle \cite{cheng2007} and others \cite{szalay2015}.
Analytical formulas for these measures are known only for special
families of states, which means that a numerical approach is usually
required.

Here, we explore the dual problem of the convex roof optimization
procedure, which is shown to be a linear semi-infinite programming
(LSIP). We prove some properties of the optimization problem using
the LSIP theory and describe the pseudocode of a central cutting-plane
algorithm (CCPA) adapted to solve the dual problem. To show how this
method performs in practice, we implement the algorithm in the MATLAB
language and calculate the multipartite entanglement quantifiers for
two families of three-qubit states. The selected measures are the
three-tangle and the $\pi$-tangle, both entanglement monotones that
quantify genuine tripartite entanglement. We choose a mixture of GHZ
and W states as one of the families, and the generalized Werner states,
a rank eight class of states, as the other one. Finally, we numerically
calculate the quantifiers and compare them with analytical values
available in the literature.

\section{Basic concepts}

\subsection{three-tangle}

Concurrence - an entanglement measure for the state $\rho$ of two
qubits - is defined as $\mathcal{C}(\rho)\equiv\max\{0,\lambda_{1}-\lambda_{2}-\lambda_{3}-\lambda_{4}\}$,
where $\lambda_{1},\dots,\lambda_{4}$ are the eigenvalues of $\sqrt{\sqrt{\rho}\tilde{\rho}\sqrt{\rho}}$
in decreasing order and $\tilde{\rho}\equiv(\sigma_{y}\otimes\sigma_{y})\rho^{\ast}(\sigma_{y}\otimes\sigma_{y})$,
with $\sigma_{y}$ as a Pauli spin matrix \cite{hill1997}. To make
the notation more economical, we symbolize a pure state of set of
density matrices $\Omega$ as $[\psi]\equiv|\psi\rangle\langle\psi|$,
where $|\psi\rangle$ is a normalized vector of the Hilbert space
$\mathcal{H}$ of the system. Then, for state $|\psi\rangle$ of three
qubits with partition $A(BC)$, the concurrence is defined as $\mathcal{C}_{A(BC)}([\psi])\equiv\sqrt{2(1-\tr\circ\tr_{BC}^{2}([\psi]))}$.
The three-tangle $\tau$ is then defined as $\tau([\psi])\equiv\left(\mathcal{C}_{A(BC)}^{2}-\mathcal{C}^{2}\circ\tr_{C}-\mathcal{C}^{2}\circ\tr_{B}\right)([\psi])$
\cite{coffman2000}. It is a measure of genuine three-qubit entanglement
and it is defined, for mixed states, as the convex roof $\tau$ in
relation to the set of pure states $\mathcal{E}$: 
\begin{equation}
\tau(\rho)\equiv\inf_{\{p_{k},|\psi_{k}\rangle\}}\sum_{k}p_{k}\tau([\psi_{k}]),
\end{equation}
such that $\sum_{k}p_{k}[\psi_{k}]=\rho$, where $\sum_{k}p_{k}=1$,
$p_{k}\geq0$, $[\psi_{k}]\in\mathcal{E}$ and the infimum is taken
over all possible pure state decompositions of $\rho$.

The three-tangle has analytical expressions for some families of states,
for example, the families $\rho_{p}\equiv p[GHZ]+(1-p)[W]$ \cite{lohmayer2006,eltschka2008}
and $\rho_{p}^{'}\equiv p[GHZ]+(1-p)\mathbb{1}/8$ \cite{siewert2012},
where $|GHZ\rangle\equiv(|000\rangle+|111\rangle)/\sqrt{2}$ and $|W\rangle\equiv(|001\rangle+|010\rangle+|100\rangle)/\sqrt{3}$
(namely GHZ and W states, respectively). Formulas are available in
Appendix \ref{families}.

\subsection{$\pi$-tangle}

An important quantifier of genuine three-qubit entanglement for pure
states is called $\pi$-tangle, or three-$\pi$ \cite{cheng2007}.
It is based on negativity \cite{vidal2002}, an entanglement monotone
given by $\mathcal{N}_{AB}(\rho)\equiv\|\rho^{T_{A}}\|_{1}-1$, where
$\|\ldotp\|_{1}$ is the trace norm and the multiplicative constant
``$1/2$'' has been removed. Let $|\psi\rangle\in\mathcal{H}_{ABC}$
be a state of a three-qubit system $ABC$ and $\pi_{A}(|\psi\rangle)\equiv\mathcal{N}_{A(BC)}^{2}([\psi])-\mathcal{N}_{AB}^{2}(\rho_{AB})-\mathcal{N}_{AC}^{2}(\rho_{AC})$,
where $\mathcal{N}_{A(BC)}([\psi])=\|[\psi]^{T_{A}}\|_{1}-1$, $\rho_{AB}\equiv\tr_{C}([\psi])$
and $\rho_{AC}\equiv\tr_{B}([\psi])$. Functions $\pi_{B}$ and $\pi_{C}$
are defined analogously. The $\pi$-tangle quantifier is then defined
as $\pi(|\psi\rangle)\equiv(\pi_{A}(|\psi\rangle)+\pi_{B}(|\psi\rangle)+\pi_{C}(|\psi\rangle))/3$.

It is proved in \cite{cheng2007} that $\pi$ is an entanglement
monotone and vanishes for product state vectors, qualifying it as
a measure of entanglement \cite{vedral1997}. It is an upper bound
on the three-tangle: $\pi(|psi\rangle)\geq\tau(|\psi\rangle)$, implying
that it is strictly positive for the states of the $GHZ\backslash W$
class. Moreover, it is also strictly positive for states of the $W$
class with the form $|\psi\rangle=\alpha|100\rangle+\beta|010\rangle+\gamma|001\rangle$
and, according to numerical calculations \cite{cheng2007}, this
is valid for other $W\backslash B$ class states, where $B$ is the
set of biseparable pure states.

\subsection{Convex roof duality}

Before talking about the dual problem of the convex roof procedure,
let us introduce some notation and definitions. Let $\mathbb{R}^{(\mathcal{E})}$
be the set of all functions $f:\mathcal{E}\rightarrow\mathbb{R}$
such that $\supp(f)<\infty$. This is a kind of ``generalized sequence''
space, with only finite ``sequences'' of real numbers indexed by
the set $\mathcal{E}$. The set $\mathbb{R}^{(\mathcal{E})}$ is a
vector subspace of the space of real functions with $\mathcal{E}$
as domain. Defining $E$ as a non-negative continuous entanglement
monotone for pure states, its convex roof $E^{\cup}$ is given by
the optimization problem $P$: 
\begin{align}
\min_{f\in\mathbb{R}^{(\mathcal{E})}}\hspace{0.5cm} & \sum_{[\psi]\in\mathcal{E}}f([\psi])E([\psi]),\nonumber \\
\text{subject to}\hspace{0.5cm} & \sum_{[\psi]\in\mathcal{E}}f([\psi])[\psi]=\rho,\quad f\geq0.
\end{align}
It is known that the Lagrangian dual problem of $P$ is given by $D$
\cite{brandao2005,eisert2007}: 
\begin{align}
\sup_{X\in H}\hspace{0.5cm} & -\tr(\rho X),\nonumber \\
\text{subject to}\hspace{0.5cm} & E([\psi])+\tr([\psi]X)\geq0,\hspace{0.1cm}\forall[\psi]\in\mathcal{E},
\end{align}
with $H$ as the space of $n$-dimensional Hermitian matrices. As
$D$ has linear objective function, a finite number of variables (setting
a base in $H$, we have $n^{2}$ real variables) and an infinite number
os linear inequalities, the problem is an LSIP \cite{reemtsen1998}.

\section{The LSIP approach}

\subsection{Properties of $P$ and $D$}

We are going to reformulate the problem $D$ according to the eigendecomposition
of $\rho=\sum_{k=1}^{r}\lambda_{k}[\phi_{k}]$, where $r$ is the
rank of $\rho$. If $|\phi_{1}\rangle,\ldots,|\phi_{r}\rangle$ are
orthonormal eigenvectors of $\rho$, then any other pure state decomposition
$\rho=\sum_{l}p_{l}[\psi_{l}]$ satisfies $|\psi_{l}\rangle\in\mathcal{H}_{\rho}\equiv\spann\{|\phi_{1}\rangle,\ldots,|\phi_{r}\rangle\},\forall l$,
where $\spann(S)$ is the linear span of the set $S$. Defining $\mathcal{E}_{\rho}\equiv\{[\psi]\in\mathcal{E}:|\psi\rangle\in\mathcal{H}_{\rho}\}$
and $H_{\rho}$ as the set of Hermitian operators on $\mathcal{H}_{\rho}$,
the reformulation of $D$ is then given by $D_{\rho}$: 
\begin{align}
-\inf_{X\in H_{\rho}}\hspace{0.5cm} & \tr(\rho X),\nonumber \\
\text{subject to}\hspace{0.5cm} & E([\psi])+\tr([\psi]X)\geq0,\hspace{0.1cm}\forall[\psi]\in\mathcal{E}_{\rho}.
\end{align}
The linear semi-infinite system of $D_{\rho}$ is defined as $\sigma_{\rho}\equiv\{E([\psi])+\tr([\psi]X)\geq0,\psi\in\mathcal{E}_{\rho}\}$,
where the set of solutions of $\sigma_{\rho}$ is the feasible set
$F_{\rho}$. As $\mathcal{E}_{\rho}$ is a compact metric space and
$E$ is continuous, $\sigma_{\rho}$ is a continuous system. Since
$E$ is non-negative, for any positive-definite matrix $X$, $E([\psi])+\tr([\psi]X)>0,\forall[\psi]\in\mathcal{E}_{\rho}$,
implying that Slater's condition is satisfied for $\sigma_{\rho}$.
So, $\sigma_{\rho}$ is a Farkas-Minkowski system and we conclude
that there is no duality gap between $D_{\rho}$ and $P$, which means
that the optimal values $v(D_{\rho})$ and $v(P)$ are equal \cite{reemtsen1998}.

The first-moment cone of $\sigma_{\rho}$ is given by $M_{c}\equiv\cone(\mathcal{E}_{\rho})$
\cite{goberna1998}, which is the set of all conical combinations
of elements of $\mathcal{E}_{\rho}$. As $\cone(\mathcal{E}_{\rho})$
is the cone of positive semi-definite matrices of $H_{\rho}$, we
have that $\inter(M_{c})\neq\emptyset$. As $\rho$ is a full rank
matrix of $H_{\rho}$, then $c\in\inter(M_{c})$. By Theorem 8.1 of
\cite{goberna1998}, we conclude that there exists an optimal solution
$X_{\rho}^{\ast}$ of $D_{\rho}$ and that the set of all optimal
solutions $F_{\rho}^{\ast}$ is bounded. Furthermore, by the same
theorem, we could conclude the absence of the duality gap without
making the continuity hypothesis on $E$. An example of discontinuous
entanglement monotone is the Schmidt number \cite{terhal2000}, which
can be defined for mixed states by means of the convex roof procedure
and, therefore, can be calculated by the dual problem.

Problem $D_{\rho}$ has several known optimality conditions, many
of which are described by Theorem 7.1 of \cite{goberna1998}. One
of them is the Karush--Kuhn--Tucker sufficient condition: $\rho\in A(X)$,
where $A(X)\equiv\cone(\mathcal{E}_{\rho}(X))$ is the cone of active
constraints and $\mathcal{E}_{\rho}(X)\equiv\{[\psi]\in\mathcal{E}_{\rho}:E([\psi])+\tr([\psi]X)=0\}$
is the set of active indexes. Since $\tr(\rho)=1$, this condition
can be reformulated as $\rho\in\conv(\mathcal{E}_{\rho}(X))$, where
$\conv(\mathcal{E}_{\rho}(X))$ is the convex hull of $\mathcal{E}_{\rho}(X)$,
which is the global optimality condition described in \cite{lee2012,ryu2012}.

There is a relationship between feasible and optimal points of $D_{\rho}$
and the so-called entanglement witnesses \cite{brandao2005}. An
entanglement witness $Y$ is a Hermitian operator, which is not positive
semidefinite z satisfying $\tr(\rho_{\text{sep}}Y)\geq0$ for any
separable $\rho_{\text{sep}}$ \cite{horodecki2001}. For the definition
of an optimal witness, we can use a bounded set $\mathcal{C}$ and
define $\mathcal{M}\equiv\cl(\mathcal{W}\cap\mathcal{C})$, where
$\cl(\mathcal{W}\cap\mathcal{C})$ is the closure of $\mathcal{W}\cap\mathcal{C}$
and $\mathcal{W}$ is the set of all entanglement witnesses. Then,
an entanglement witness $Y^{\ast}$ is $\rho$-optimal if $\tr(\rho Y^{\ast})=\min_{Y\in\mathcal{M}}\tr(\rho Y)$
\cite{terhal2002,brandao2005}. If $E([\psi])=0$ for any separable
state $|\psi\rangle$ (if this is not true, there exists $E_{0}\in\mathbb{R}$
such that $E+E_{0}$ is a non-negative entanglement monotone), it
can be verified than any optimal solution $X_{\rho}^{\ast}\neq0$
is a $\rho$-optimal entanglement witness (the set $\mathcal{C}$
can be any bounded set such that $F_{\rho}^{\ast}\subseteq\mathcal{C}$).
In addition, any feasible $X$ such that $\tr(\rho X)<0$ is an entanglement
witness.

\subsection{Central cutting-plane algorithm}

To numerically solve an LSIP problem, several methods are available,
mostly classified into five categories: discretization methods, local
reduction methods, exchange methods, simplex-like methods and descent
methods, ordered in decreasing order of efficiency according to reference
\cite{goberna1998}. Besides these approaches, other deterministic
types of algorithms and uncertain LSIP methods are discussed in a
recent review of the field \cite{goberna2017}. We then choose the
CCPA \cite{gribik1979} to tackle problem $D_{\rho}$, which is classified
as a discretization method. For the sake of simplicity, we work with
the first version of the algorithm, while subsequent improvements
are found in Part IV in \cite{goberna1998} and in \cite{betro2004}.
The CCPA has the advantage of having the property of global convergence,
unlike the reduction procedure and almost all methods based on the
primal problem $P$, such as the usual algorithms implemented for
calculating the convex roof \cite{zyczkowski1999,audenaert2001,rothlisberger2009,cao2010,rothlisberger2012}.
Also, it generates a sequence of feasible points that converges to
an optimal value or to a limit point of an optimal value, implying
that a convergent sequence of lower bounds is generated.

In order to successfully employ the CCPA, some conditions need to
be satisfied by $D_{\rho}$. According to \cite{gribik1979}, we
need to restrict the feasible set $F_{\rho}$ to the set $F_{\rho}^{\prime}\equiv F_{\rho}\cap\mathcal{C}$,
where $\mathcal{C}\subset H_{\rho}$ is a compact convex set. Since
$F_{\rho}^{\ast}$ is bounded, there exists $\delta>0$ such that
$F_{\rho}^{\ast}\subseteq B_{\delta}$, where $B_{\delta}\equiv\{X\in H_{\rho}:\|X\|\leq\delta\}$,
with $\|\ldotp\|$ as the operator norm, is a compact convex set.
A valid value of $\delta$ is provided for normalized measures ($0\leq E([\psi])\leq1,\forall[\psi]\in\mathcal{E}_{\rho}$)
by Lemma \ref{bounded} in Appendix \ref{proof:lemma}. Other non-trivial
conditions are the existence of a non-optimal Slater point (a point
that satisfy the Slater's condition) $X$, which is clearly satisfied,
and the continuity of $E$. However, to make the optimization problem
easier to solve, we choose an orthonormal basis $\{Z_{1},\ldots,Z_{r^{2}}\}$
of $H_{\rho}$, use the result of Corollary \ref{bounded_real} and
replace the problem $D_{\rho}$ by the problem $D_{c}$: 
\begin{align}
-\inf_{x\in\mathbb{R}^{r^{2}}}\hspace{0.5cm} & \langle c,x\rangle,\nonumber \\
\text{subject to}\hspace{0.5cm} & \tilde{E}(\psi)+\langle\psi,x\rangle\geq0,\hspace{0.1cm}\forall\psi\in\tilde{\mathcal{E}}_{c},\nonumber \\
 & |x_{m}|\leq r(r-1)\frac{\lambda_{r}}{\lambda_{1}},\hspace{0.1cm}1\leq m\leq r^{2},
\end{align}
where $X=\sum_{k}x_{k}Z_{k}$, $\rho=\sum_{k}c_{k}Z_{k}$, $\psi\equiv(\psi_{1},\ldots,\psi_{r^{2}})$,
$x\equiv(x_{1},\ldots,x_{r^{2}})$, $c\equiv(c_{1},\ldots,c_{r^{2}})$,
$\tilde{\mathcal{E}}_{c}\equiv\{\psi\in\mathbb{R}^{r^{2}}:\sum_{k}\psi_{k}Z_{k}\in\mathcal{E}_{\rho}\}$
and $\tilde{E}(\psi)\equiv E([\psi])$. To simplify the discussion
of the CCPA, we omit the deletion rules in the pseudocode present
in \cite{gribik1979}, as they are not necessary for the convergence
of the algorithm. The pseudocode of the CCPA in \cite{gribik1979},
for a tolerance $\epsilon>0$, is given by following steps: 
\begin{itemize}
\item[Step 0:] Let $\bar{E}$ be strictly greater than $-v(D_{c})$. Let $SD_{c}^{0}$
be the program: 
\begin{align}
\max_{(y,x)\in\mathbb{R}^{r^{2}+1}}\hspace{0.5cm} & y,\nonumber \\
\text{subject to}\hspace{0.5cm} & \langle c,x\rangle+y\|c\|_{2}\leq\bar{E},\nonumber \\
 & |x_{m}|\leq r(r-1)\frac{\lambda_{r}}{\lambda_{1}},\hspace{0.1cm}1\leq m\leq r^{2}.
\end{align}
Choose $w^{(0)}\in\mathbb{R}^{r^{2}}$ such that $|w^{(0)}|\leq r(r-1)\lambda_{r}/\lambda_{1},\hspace{0.1cm}1\leq m\leq r^{2}$.
Let $k=1$. 
\item[Step 1:] Let $(x^{(k)},y^{(k)})\in\mathbb{R}^{r^{2}+1}$ be a solution of
$SD_{c}^{k-1}$. If $|y|<\epsilon$, stop. Otherwise, go to Step 2. 
\item[Step 2:] 
\begin{itemize}
\item[(i)] If $v(D_{\text{aux}}^{k})\geq0$, where $D_{\text{aux}}^{k}:\inf_{\psi\in\tilde{\mathcal{E}}_{c}}\tilde{E}(\psi)+\langle\psi,x^{(k)}\rangle$,
add the constraint $\langle c,x\rangle+y\|c\|_{2}\leq\langle c,x^{(k)}\rangle$
to program $SD_{c}^{k-1}$. Set $w^{(k)}=x^{(k)}$. 
\item[(ii)] Otherwise, add the constraint $\langle\psi^{(k)},x\rangle-y\|\psi^{(k)}\|_{2}\geq-\tilde{E}(\psi^{(k)})$
to program $SD_{c}^{k-1}$. Set $w^{(k)}=w^{(k-1)}$. 
\end{itemize}
In either case, call the resulting program $SD_{c}^{k}$. Set $k=k+1$
and return to step 1. 
\end{itemize}
By Lemma 1 of \cite{gribik1979} and the tolerance $\epsilon$ in
Step 1, the algorithm always terminates. Furthermore, by Theorem 1
of \cite{gribik1979} and dropping the tolerance requirement in Step
1, the sequence of feasible points $\{w^{(k)}\}_{k=0}^{\infty}$ has
limit points and they are optimal, which is the property of global
convergence.

\section{Numerical calculations for $\pi$-tangle and three tangle}

\subsection{Convex roof of $\pi$-tangle}

As an entanglement monotone for pure states \cite{cheng2007}, the
convex roof of the $\pi$-tangle is also an entanglement monotone
\cite{Vidal2000}. Also, it vanishes for any biseparable state, it
is nonzero for $GHZ\backslash W$ class and it is nonzero for at least
some states of the $W\backslash B$ class, which includes any mixed
state with generalized W states \cite{eltschka2008} and biseparable
states in its optimal decomposition. To date, the $\pi$-tangle for
mixed states has been calculated only for the mixture $\rho_{p}$
of W and GHZ states \cite{ma2013}, which has the analytical result
described in Appendix \ref{families}. Here, we numerically reproduce
their analytical result and also calculate it for $\rho'_{p}$.

\subsection{Numerical results}

\label{numerical_results}

We implement the CCPA pseudocode in MATLAB scripts for numerical calculations.
The two main procedures of the algorithm are the linear and nonlinear
optimization problems $SD_{c}^{0}$ and $D_{\text{aux}}^{k}$, respectively,
which were implemented by the linprog function and by the GlobalSearch
object. We use the code to calculate the $\pi$-tangle and the three-tangle
for two families of states: $\rho_{p}$ and $\rho'_{p}$. We also
compare the numerical calculations with the available analytical formulas
in Appendix \ref{families}. The results for the states $\rho_{p}$
are expressed in Fig. \ref{GHZ_W}, which show good agreement with
the analytical curves. With a tolerance $\epsilon=10^{-3}$, we achieve
the results in few minutes using a common notebook {[}\textbf{Processador
e memória?}{]}. For the states $\rho'_{p}$, using $\epsilon=10^{-5}$,
the numerical three-tangle is slightly lower than the exact nonzero
values, according to Fig. \ref{Werner}. This agrees with the fact
that the CCPA gives a lower bound on the convex roof when it finds
a feasible suboptimal solution. As the CCPA has global convergence,
one can generate a larger sequence of feasible points that gives values
closer to the exact one. For $\rho'_{p}$, sequences of no more than
12 feasible solutions were generated for each value of $p$ and each
calculation spent few hours. Since $\rho'_{p}$ is a rank 8 family
of states, this higher computational cost is justified as $\rho_{p}$
has only rank 2. Furthermore, both quantifiers spend a similar amount
of calculation time for each state.

\begin{figure}[h]
\includegraphics[width=1\columnwidth]{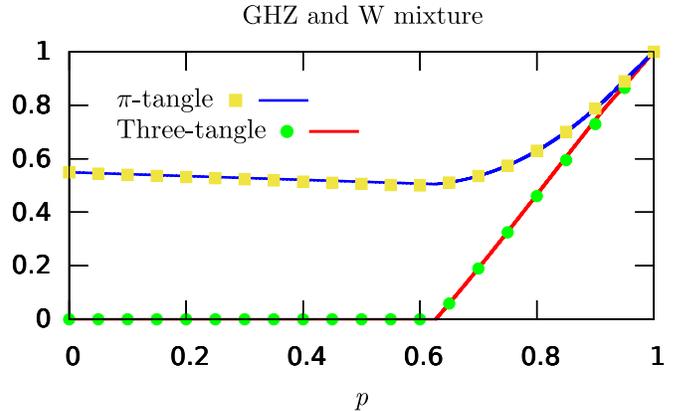} 
\caption{Three-tangle and $\pi$-tangle calculated for states $\rho_{p}$.
Symbols (boxes and circles) and continuous lines represent numerical
and analytical values, respectively.}
\label{GHZ_W} 
\end{figure}

\begin{figure}[h]
\includegraphics[width=1\columnwidth]{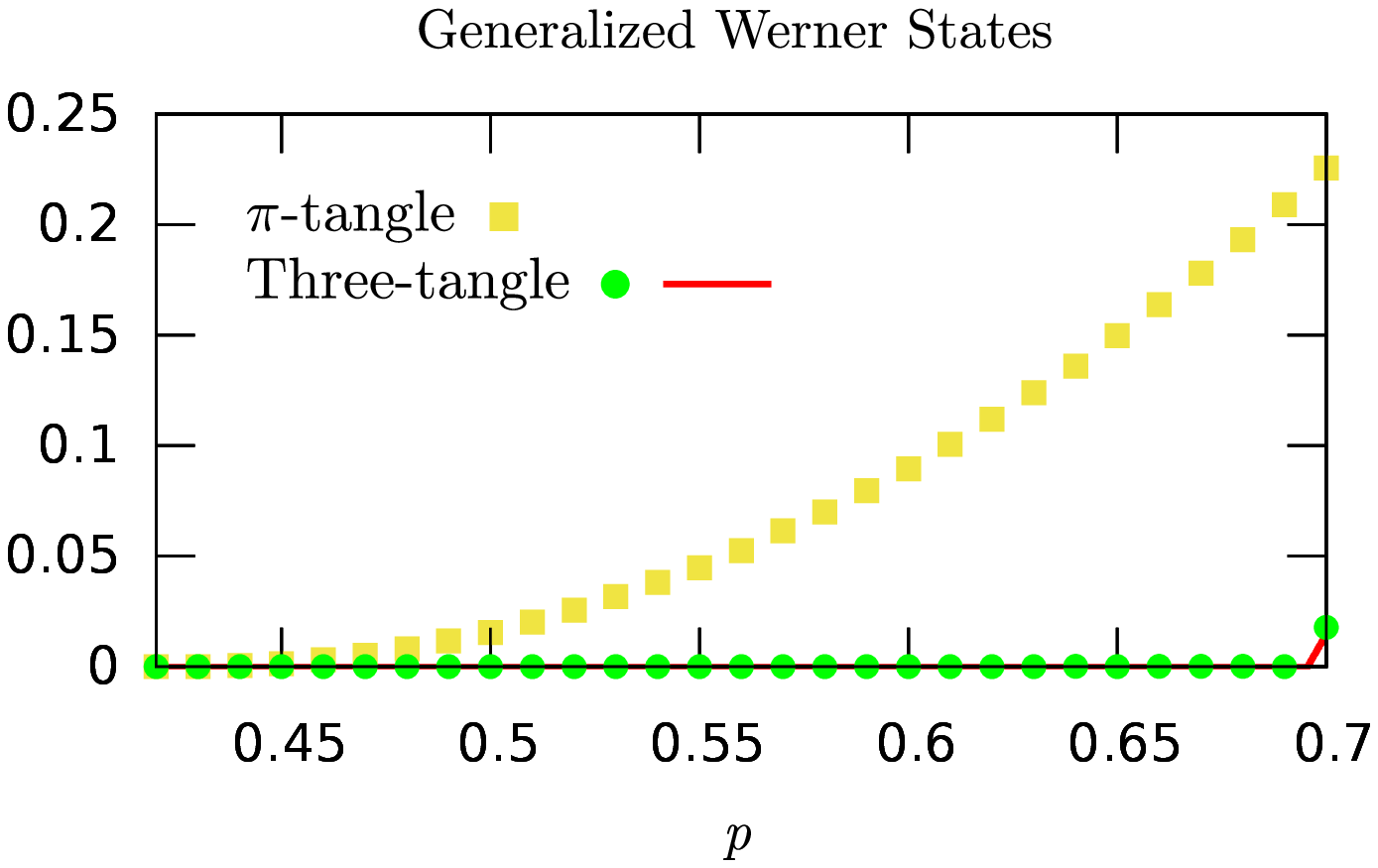}
\caption{Three-tangle and $\pi$-tangle calculated for states $\rho'_{p}$.
Symbols (boxes and circles) and the continuous line represent numerical
and analytical values, respectively.}
\label{Werner} 
\end{figure}

The three-tangle and $\pi$-tangle measures can be used to discriminate
among the classes $B$, $W\backslash B$ and $GHZ\backslash W$ \cite{cheng2007}.
For the family of states $2\rho_{p}$, the analytical results in \cite{lohmayer2006,eltschka2008},
and described in Appendix \ref{families}, show that $\rho_{p}$ belongs
to the $W\backslash B$ class for $p\lesssim0.62685$ and to the $GHZ\backslash W$
class for higher values of $p$. As show by Fig. \ref{GHZ_W}, the
positive values of the three-tangle indicate the $GHZ\backslash W$
class, whereas the positive values of the $\pi$-tangle in the region
where the three-tangle is zero show that the state belongs to the
$W\backslash B$ class. The graph around the class transition point,
calculated with a tolerance $\epsilon=10^{-5}$ and depicted in Fig.
\ref{GHZ_W_zoom}, shows that the numerical result is in a good agreement
with the analytical one. In the case of the family $\rho'_{p}$, it
belongs to the $B$ class for $p\leq p_{B}\equiv3/7\approx0.42857$,
to the $W\backslash B$ class for $3/7<p\leq p_{W}\approx0.69554$
and to the $GHZ\backslash W$ class for $p>p_{W}$ \cite{guhne2010,siewert2012-2}.
The plot in Fig. \ref{Werner_zoom} shows that the class transition
in $p_{B}$ occurs between $p=0.43$ and $p=0.44$, which is only
slightly higher than $p_{B}$, which is expected since the algorithm
gives a close lower bound to the optimal value. In addition, the numerical
values in the graph show the transition between classes $W\backslash B$
and $GHZ\backslash W$.

\begin{figure}[tb]
\includegraphics[width=1\columnwidth]{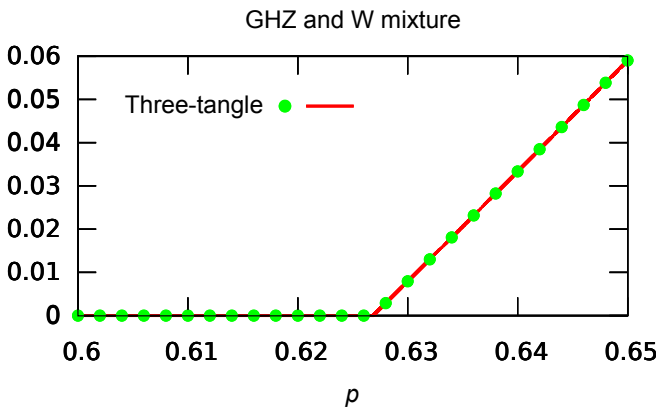} 
\caption{Three-tangle calculated for states $\rho_{p}$. Symbols (circles)
and the continuous line represent numerical and analytical values,
respectively.}
\label{GHZ_W_zoom} 
\end{figure}

\begin{figure}[tb]
\includegraphics[width=1\columnwidth]{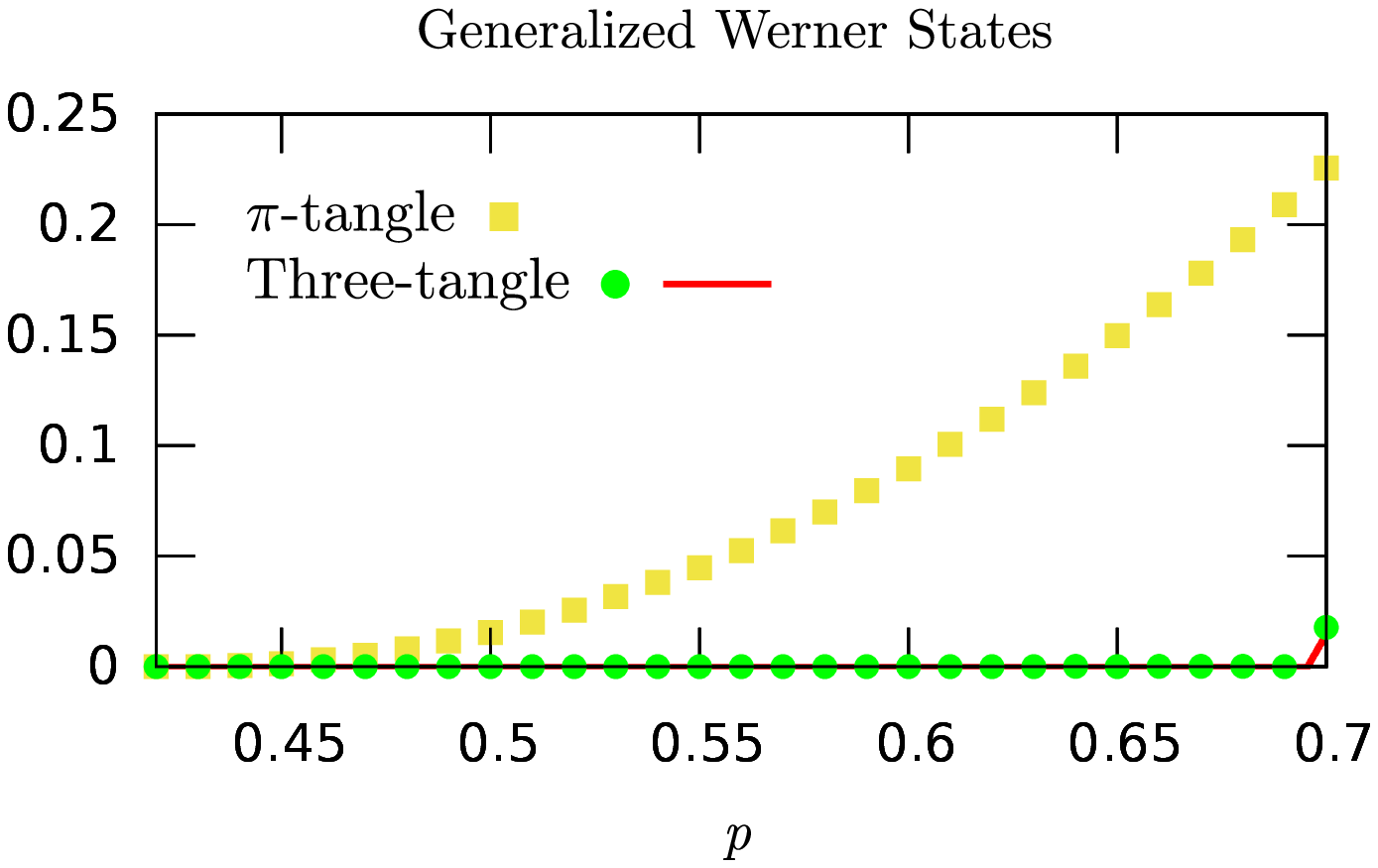}
\caption{Three-tangle and $\pi$-tangle calculated for states $\rho'_{p}$.
Symbols (boxes and circles) and the continuous line represent numerical
and analytical values, respectively.}
\label{Werner_zoom} 
\end{figure}


\section{Conclusion}

We explored the theory of LSIP to derive properties of the dual problem
of the convex roof procedure that gives entanglement monotones for
mixed states from pure state measures. We showed that the absence
of the duality gap between primal and dual problems occurs in very
general conditions. In addition, we proved that the set of optimal
points is non-empty and bounded and we derived bounds on the coefficients
of optimal solutions. For numerical calculations, we wrote the dual
problem in a suitable LSIP form and described the pseudocode of an
CCPA designed for this type of optimization. To check the performance
of the algorithm, we calculated two measures of genuine three-qubit
entanglement, three-tangle and $\pi$-tangle, for the mixture of GHZ
and W states and for the generalized Werner states, a full rank family
of states. We compared the numerical results with the available analytical
values and verified that the CCPA results are very close the exact
ones for the lower rank family of states, while providing lower bounds
for the high rank ones. As the algorithm gives lower bounds on the
amount of entanglement for suboptimal feasible points and global convergence,
the results are in agreement with the expected behavior. Furthermore,
we used the difference between the two measures to distinguish $GHZ\backslash W$
and $W$ classes, in agreement with the entanglement classification
of these states in the literature.

We believe that our work gives a good alternative to the convex roof
calculation of mixed states entanglement, especially when close lower
bounds are required. The CCPA has very general applicability, working
with discontinuous measures and multipartite states with any finite
rank. For future works, we expect to apply other LSIP algorithms to
the convex roof problem, with the necessary modifications and improvements. 
\begin{acknowledgments}
The authors acknowledge the financial support of the Brazilian agencies
CNPq (\#312723/2018-0, \#306065/2019-3 \& \#425718/2018-2), CAPES
(PROCAD - 2013), and FAPEG (PRONEM \#201710267000540, PRONEX \#201710267000503).
This work was also performed as part of the Brazilian National Institute
of Science and Technology (INCT) for Quantum Information (\#465469/2014-0). 
\end{acknowledgments}

\appendix

\section{\label{families} Thee-tangle and $\pi$-tangle for families of states}

Here, we show the analytical expressions for the three-tangle and
$\pi$-tangle for the families of states $\rho_{p}$ and $\rho'_{p}$
available in the literature. First, we show the formulas for the three-tangle
quantifier applied to the mixture of GHZ and W states: $\rho_{p}$.
Set $s\equiv8\sqrt{6}/9$, $p_{0}\equiv s^{2/3}/(1+s^{2/3})$, $p_{1}\equiv1/2+1/\big(2\sqrt{1+s^{2}}\big)$.
The three-tangle of $\rho_{p}$ is given by \cite{lohmayer2006,eltschka2008}
\begin{equation}
\tau(\rho_{p})=\Bigg\{\begin{matrix}\hspace{-1.4cm}0 & \hspace{-0.3cm}\text{for \ensuremath{p\leq p_{0}},}\\
\hspace{-0.5cm}\tau_{3}(p,0) & \hspace{0.5cm}\text{for \ensuremath{p_{0}<p\leq p_{1}},}\\
\hspace{0.1cm}\tau_{3}^{\text{conv}}(p,p_{1}) & \hspace{-0.3cm}\text{for \ensuremath{p>p_{1}},}
\end{matrix}
\end{equation}
where $\tau_{3}(p,0)\equiv\big\arrowvert p^{2}-16\sqrt{p(1-p)^{3}}/3\sqrt{6}\big\arrowvert$
and $\tau_{3}^{\text{conv}}(p,p_{1})\equiv\big[p-p_{1}+(1-p)\big(p_{1}^{2}-s\sqrt{p_{1}(1-p_{1})^{3}}\big)\big]/(1-p_{1})$.

The family of states $\rho'_{p}$, as the parameter $p$ ranges from
0 to 1, goes through all three-qubit entanglement classes: $S,B\backslash S,W\backslash B$
and $GHZ\backslash W$ \cite{siewert2012-2}, where $S$ is the class
of separable states. The value $p_{W}$ of $p$ that separates the
classes $W$ and $GHZ\backslash W$ is $p_{W}\approx0.6955427$. The
three-tangle of $\rho'_{p}$ is then given by \cite{siewert2012}
\begin{equation}
\tau(\rho'_{p})=\bigg\{\begin{matrix}\hspace{-0.5cm}0 & \text{for \ensuremath{p\leq p_{W}},}\\
\hspace{0.1cm}\frac{p-p_{W}}{1-p_{W}} & \hspace{0.65cm}\text{for \ensuremath{p_{W}<p\leq1}.}
\end{matrix}
\end{equation}

The last available analytical result is the $\pi$-tangle of the states
$\rho_{p}$, which is given by \cite{ma2013}: 
\begin{equation}
\pi(\rho_{p})=\Bigg\{\begin{matrix}\pi^{(1)}(\rho_{p}) & \text{for \ensuremath{0\leq p\leq p_{0}},}\\
\pi^{(2)}(\rho_{p}) & \text{for \ensuremath{p_{0}<p\leq p_{1}},}\\
\pi^{(3)}(\rho_{p}) & \text{for \ensuremath{p_{1}<p\leq1},}
\end{matrix}
\end{equation}
where $\pi^{(1)}(\rho_{p})\equiv\big\{4(\sqrt{5}-1)(p_{0}-p)+p\big[5p_{0}^{2}-4p_{0}+8-18\big(\sum_{i=1}^{4}|\lambda_{i}(p_{0})|-1\big)^{2}\big]\big\}/9p_{0}$,
$\pi^{(2)}(\rho_{p})\equiv\big[5p^{2}-4p+8-18\big(\sum_{i=1}^{4}|\lambda_{i}(p)|-1\big)^{2}\big]/9$
and $\pi^{(3)}(\rho_{p})\equiv\big\{ p-p_{1}+(1-p)\big[5p_{1}^{2}-4p_{1}+8-18\big(\sum_{i=1}^{4}|\lambda_{i}(p_{1})|-1\big)^{2}\big]/9\big\}/(1-p_{1})$.
For a fixed value of $p$, each $\lambda_{i}(p)$, for $i\in\{1,\ldots,4\}$,
is a solution of the following equation: 
\begin{align*}
 & \lambda^{4}-\lambda^{3}+\left(\frac{5}{36}p^{2}-\frac{p}{9}+\frac{2}{9}\right)\lambda^{2}+\bigg[\frac{(p(1-p))^{3/2}}{3\sqrt{6}}\\
 & -\frac{7}{27}p^{3}+\frac{7}{18}p^{2}-\frac{p}{6}+\frac{1}{27}\bigg]\lambda+\bigg[-\frac{p(p(1-p))^{3/2}}{6\sqrt{6}}\\
 & -\frac{41}{648}p^{4}+\frac{149}{648}p^{3}-\frac{13}{54}p^{2}+\frac{7}{81}p-\frac{1}{81}\bigg]=0.
\end{align*}

\section{\label{proof:lemma} Bounding the feasible set}

\begin{lemma} \label{bounded} If $0\leq E([\psi])\leq1,\forall[\psi]\in\mathcal{E}_{\rho},$
and $\delta\equiv(r-1)\lambda_{r}/\lambda_{1}$, where $r=\rank(\rho)$,
$\lambda_{1}$ and $\lambda_{r}$ are the lowest and highest eigenvalues
of $\rho$, respectively, then $F_{\rho}^{\ast}\subseteq B_{\delta}$.
\end{lemma} \begin{proof} Let $x_{1}\leq\ldots\leq x_{r}$ be the
eigenvalues of $X\in H_{\rho}$. By the constraint $E([\psi])+\tr([\psi]X)\geq0,\forall[\psi]\in\mathcal{E}_{\rho}$,
of the problem $D_{\rho}$ and the min-max theorem, we have that $x_{1}\geq-1$
is a necessary condition for the feasibility of $X$. Let $\{|x_{1}\rangle,\ldots|x_{r}\rangle\}$
be an orthonormal basis with eigenvectors of $X$ and $\rho=\sum_{k,l}\lambda_{k,l}^{\prime}|x_{k}\rangle\langle x_{l}|$.
As $x_{1}\geq-1$, $0\leq E^{\cup}(\rho)\leq1$ and by the fact that
there is no duality gap between $D_{\rho}$ and $P$, if $X\in F_{\rho}^{\ast}$
then 
\begin{equation}
\tr(\rho X)=\sum_{k}\lambda_{k,k}^{\prime}x_{k}\leq0\Rightarrow x_{r}\leq(r-1)\frac{\lambda_{r}}{\lambda_{1}}.\label{eq_lemma}
\end{equation}
Equation \ref{eq_lemma} implies that $\|X\|=\sup\{\|X|\psi\rangle\|_{2}:\||\psi\rangle\|_{2}=1\}=\max\{|x_{1}|,|x_{r}|\}\leq(r-1)\lambda_{r}/\lambda_{1}$.
Thus, we conclude that $F_{\rho}^{\ast}\subseteq B_{\delta}$ for
$\delta\equiv(r-1)\lambda_{r}/\lambda_{1}$. \end{proof}

\begin{corol} \label{bounded_real} Let $\{Z_{1},\ldots,Z_{r^{2}}\}$
be an orthonormal basis of $H_{\rho}$, $r=\rank(\rho)$ and $\lambda_{1}\leq\ldots\leq\lambda_{r}$
the eigenvalues of $\rho$. If $X\in F_{\rho}^{\ast}$ then $x_{m}^{\prime}\equiv|\tr(A_{m}X)|\leq r(r-1)\lambda_{r}/\lambda_{1}$.
\end{corol} \begin{proof} Let $\{|x_{1}\rangle,\ldots|x_{r}\rangle\}$
be an orthonormal basis with eigenvectors of $X$ and $X=\sum_{k}x_{k}[x_{k}]$.
By Lemma \ref{bounded}, $|\tr(Z_{m}X)|\leq\sum_{k}|x_{k}||\tr(Z_{m}[x_{k}])|\leq\sum_{k}|x_{k}|\leq r(r-1)\lambda_{r}/\lambda_{1}$.
\end{proof}

\bibliographystyle{apsrev4-1}
\bibliography{Refs}

\end{document}